\documentclass[11pt]{article}
\usepackage{amssymb} % set of additional math symbols
\usepackage{amsmath}     % new environments
\usepackage{amsthm}
\usepackage{hyperref}

\usepackage{a4wide}
\usepackage{graphicx}               % eps picture

\usepackage{algorithm,algorithmic}

\newtheorem{thm}{Theorem}
\newtheorem{lem}{Lemma}

\newtheorem{cor}{Corollary}

\newtheorem{prop}{Proposition}
\newtheorem{obs}{Observation}

\newenvironment{keyword}{\par{\noindent\bf Keywords:}}

\begin{document}

%*******TITLE AND AUTHORS*******************************************
\title{Single machine scheduling problems with uncertain parameters and the OWA criterion}

\author{Adam Kasperski\\ % the first author
  {\small \textit{Institute of Industrial}}\\
  {\small \textit{Engineering and Management,}}\\
  {\small \textit{Wroc{\l}aw University of Technology,}}\\
  {\small \textit{Wybrze{\.z}e Wyspia{\'n}skiego 27,}}\\
  {\small \textit{50-370 Wroc{\l}aw, Poland}}\\
  {\small \textit{adam.kasperski@pwr.edu.pl}}
  \and
  Pawe{\l} Zieli{\'n}ski\\   % the second author
  {\small \textit{Institute of Mathematics}}\\
  {\small \textit{and Computer Science}}\\
  {\small \textit{Wroc{\l}aw University of Technology,}}\\
  {\small \textit{Wybrze{\.z}e Wyspia{\'n}skiego 27,}}\\
  {\small \textit{50-370 Wroc{\l}aw, Poland}}\\
  {\small \textit{pawel.zielinski@pwr.edu.pl}}}

%\date{}
\maketitle

%********************ABSTRACT*****************************************

\begin{abstract}
		In this paper a class of single machine scheduling problems is discussed. It is assumed that  job parameters, such as processing times, due dates, or weights are uncertain and their values are specified in the form of a discrete scenario set. The Ordered Weighted Averaging (OWA) aggregation operator is used to choose an optimal schedule. The OWA operator generalizes traditional criteria used in decision making under uncertainty, such as the maximum, average, median or Hurwicz criterion. It also allows us to extend the robust approach to scheduling by taking into account various attitudes of decision makers towards a risk. In this paper a general framework for solving single machine scheduling problems with the OWA criterion is proposed and some positive and negative computational results for two basic single machine scheduling problems are provided.
\end{abstract}

%%********KEYWORDS*****************************************************

\section{Introduction}

Scheduling under uncertainty is an important and extensively studied area of operations research and discrete optimization. The importance of this research direction results from the fact that in many real-world problems the precise values of parameters in scheduling models are not  known in advance. Thus, instead of possessing the exact values of the parameters, decision makers have rather a set of all their possible realizations, called a \emph{scenario set}. In some cases an additional information with this scenario set is available. If   a probability distribution  in the scenario set is known, then 
\emph{stochastic approach} can be used, which typically consists in minimizing the expected solution cost (see, e.g.~\cite{P02}). The unknown probability distribution can be upper bounded by
 a \emph{possibility distribution}, which leads to  \emph{possibilistic} (fuzzy) scheduling problems
  (see, e.g~\cite{KZ11f}). Finally, if no additional information with scenario set is provided,
   then \emph{robust approach} is usually used (see, e.g.~\cite{KY97}). In the robust optimization, we seek a solution minimizing a cost in the worst case, which usually leads  to applying the  
    \emph{minmax} or  \emph{minmax regret} criterion for choosing a solution.

The robust approach to decision making is often regarded as too conservative or pessimistic. It follows from the fact, that the minmax criterion takes only the worst-case scenarios into account, ignoring the information connected with the remaining scenarios. This criterion also assumes that decision makers are very risk averse, which is not always true. These drawbacks of the minmax criterion are well known in decision theory, and a detailed discussion on this topic can be found, for example,  in~\cite{LR57}. 
In this paper we will assume that a scenario set associated with scheduling problem is specified by enumerating all possible scenarios. Such a representation of scenario sets is called a \emph{discrete uncertainty representation} and has been described, for instance, in~\cite{KY97}.
Our goal is to generalize the minmax approach to scheduling problems under uncertainty by using the \emph{Ordered Weighted Averaging} aggregation operator (OWA for short) introduced by Yager in~\cite{YA88}.
The OWA operator is widely applied to aggregate criteria in multiobjective decision problems (see, e.g.,~\cite{GS12, YKB11, OS03}) but it can also be applied to choose a solution under the discrete uncertainty representation by identifying scenarios with objectives in a natural way. 
The OWA operator generalizes the classical criteria used in decision making under uncertainty such as the maximum, minimum, average, median, or Hurwicz criterion~\cite{LR57}. So, by using OWA we can extend the minmax approach, typically used in the
robust optimization. Furthermore, the weights used in the OWA operator allows us to model various attitudes of decision makers towards a risk.

Since we generalize the minmax approach to single machine scheduling problems under the 
discrete uncertainty representation, let us briefly recall the known results in this area (see also~\cite{KZ14} for a survey). 
The minmax version of the single machine scheduling problem with the total flow time criterion has been studied  in~\cite{YY02}, where it has been shown that the problem is NP-hard even for two processing time scenarios and strongly NP-hard when the number of processing time scenarios is a part of the
input (the unbounded case). A generalization of this problem, with the weighted sum of completion times criterion, has been recently discussed in~\cite{MNO13, FA10} where, in particular,  several inapproximability results for that problem have been established. We will describe these results in more detail later in this paper. In~\cite{AC08} the minmax version of the single machine scheduling problem with the maximum weighted tardiness criterion has been discussed, where it has been shown that some special cases of the problem are polynomially solvable. In this paper, we generalize and extend the algorithms proposed in~\cite{AC08}. In~\cite{ AAK11,AC08} the minmax version of the single machine scheduling problem with the number of late jobs criterion has been investigated. It has been shown in~\cite{AC08} that the problem is NP-hard for deterministic due dates and two processing time scenarios. On the other hand, it has been shown in~\cite{AAK11} that the problem with unit processing times and 
the number of due date scenarios being a part of the input is strongly NP-hard and hard to approximate within a factor less than~2.
In a more general version of this problem the weighted sum of late jobs is minimized.  This problem is known to be NP-hard for two weight scenarios~\cite{AV01}, strongly NP-hard and hard to approximate within any constant factor if the number of weight scenarios is a part of the input~\cite{KZ13}.

This paper is organized as follows. Section~\ref{sec1} presents
 a formulation of the general problem under consideration  as well as some its special cases. 
 The next two sections discuss two basic single machine scheduling problems. Namely,  Section~\ref{sec2} explores the problem with the maximum weighted tardiness cost function and Section~\ref{sec3} investigates the problem in which the cost function is the weighted sum of completion times. We show that both problems have various computational properties which depend on the weight distribution in the OWA operator. For some weight distributions the problems are polynomially solvable, while for other ones they become strongly NP-hard and  are also hard to approximate.

\section{Problem formulation}
\label{sec1}

Let $J=\{J_1,\dots,J_n\}$ be a set of jobs which must be processed on a single machine.  For simplicity of notations, we will identify job $J_j$ with its index~$j$. The set of jobs may be partially ordered by some precedence constraints. The notation $i\rightarrow j$ means that processing of job $j$ cannot start before processing of job $i$ is completed  (job $j$ is called a \emph{successor} of job $i$). 
  For each job $j$ the following parameters may be specified: a nonnegative \emph{processing time} $p_j$, a nonnegative \emph{due date} $d_j$ and a nonnegative \emph{weight} $w_j$. The due date $d_j$ expresses a desired completion time of $j$ and the weight $w_j$ expresses the importance of job~$j$ relative to the other jobs in the system.  In all scheduling models discussed in this paper we assume that all the jobs are ready for processing at time~0, in other words, each job has a release date equal to~0. We also assume that each job must be processed without any interruptions, so we consider only nonpreemptive models. Under these assumptions we can define a \emph{schedule} $\pi$ as a feasible permutation of the jobs, in which the precedence constraints among the jobs are preserved. The set of all feasible schedules will be denoted by~$\Pi$. 
  
  Let us denote by $C_j(\pi)$ the completion time of job $j$ in schedule $\pi$. We will use $f(\pi)$ to denote a cost of schedule $\pi$. The value of $f(\pi)$ depends on job completion  times and may also depend on job due dates or weights. In this paper we will investigate two 
   basic scheduling problems, in which the cost function is the maximum weighted tardiness, i.e. $f(\pi)=\max_{j\in J} w_j[C_j(\pi)-d_j]^+$ (we use the notation $[x]^+=\max\{0,x\}$) and the weighted sum of completion times, i.e. $f(\pi)=\sum_{j\in J} w_j C_j(\pi)$.  In the  deterministic case, we wish to find a feasible schedule which minimizes 
the cost $f(\pi)$, that is:
$$\mathcal{P}:\; \min_{\pi \in \Pi} f(\pi).$$

We now study a situation in which some or all problem parameters are ill-known.
Let~$S$ be a vector of the problem parameters  which may occur.  The vector $S$ is called a \emph{scenario}.  We will use $p_j(S)$, $d_j(S)$ and $w_j(S)$ to denote the processing time, due date, and weight of job $j$ under scenario $S$. A parameter is deterministic (precisely known) if its value is the same under each scenario. 
Let a \emph{scenario set} $\Gamma=\{S_1,\dots,S_K\}$ contain all possible scenarios, where $K> 1$.
In this paper, we
distinguish  the \emph{bounded case}, where~$K$ is bounded by a constant and 
the \emph{unbounded case}, where~$K$ is a part of the input.
 Now, the completion time of job~$j$ in $\pi$ and the cost of~$\pi$ depend on scenario $S_i \in \Gamma$ and will be denoted by~$C_j(\pi,S_i)$ and $f(\pi,S_i)$, respectively. 

Since scenario  set $\Gamma$ contains more than one scenario, an additional criterion is required to choose a reasonable solution.
In this paper we suggest  to use the  \emph{Ordered Weighted Averaging} aggregation operator
 (OWA for short) proposed by Yager in~\cite{YA88}. We now describe this criterion. Let $(f_1,\dots,f_K)$ be a vector of real numbers. Let us introduce a vector of weights $\pmb{v}=(v_1,\dots,v_K)$ such that $v_j\in [0,1]$ for all $j\in [K]$ ($[K]$ stands for the set $\{1,\dots,K\}$) and $v_1+\dots+v_K=1$. Let $\sigma$ be a permutation of $[K]$ such that $f_{\sigma(1)}\geq f_{\sigma(2)}\geq \dots \geq f_{\sigma(K)}$. The OWA operator is defined as follows:
$${\rm owa}_{\pmb{v}}(f_1,\dots,f_K)=\sum_{i\in[K]} v_i f_{\sigma(i)}.$$
The OWA operator has several natural properties which follow directly from its definition (see, e.g.~\cite{YKB11}). Since it is a \emph{convex combination} of the cost functions, 
$\min(f_1,\dots,f_K) \leq \mathrm{owa}_{\pmb{v}}(f_1,\dots,f_K) \leq \max(f_1,\dots,f_K)$. It is also \emph{monotonic}, i.e. if $f_j \geq g_j$ for all $j\in [K]$, then 
$\mathrm{owa}_{\pmb{v}}(f_1,\dots,f_K)\geq \mathrm{owa}_{\pmb{v}}(g_1,\dots,g_K)$,
\emph{idempotent}, i.e. if $f_1=\dots=f_k=a$, then $\mathrm{owa}_{\pmb{v}}(f_1,\dots,f_K)=a$ and 
\emph{symmetric}, i.e. its value does not depend on the order of the values $f_j$, $j\in [K]$. The OWA operator generalizes some important criteria used in decision making under uncertainty.   If $v_1=1$ and $v_j=0$ for $j=2,\dots,K$, then OWA becomes the maximum. If $v_K=1$ and $v_j=0$ for $j=1,\dots,K-1$, then OWA becomes the minimum. In general, if $v_k=1$ and $v_j=0$ for $j\in [K]\setminus\{k\}$, then OWA is the $k$th largest element among $f_1,\dots,f_K$. In particular, when $k=\lfloor K/2 \rfloor +1$, the $k$th element is the median.
If $v_j=1/K$ for all $j\in [K]$, i.e. when the  weights are \emph{uniform}, then OWA is the average (or the Laplace criterion). Finally, if $v_1=\alpha$ and $v_K=1-\alpha$ for some fixed $\alpha\in [0,1]$ and $v_j=0$ for the remaining weights, then we get the Hurwicz pessimism-optimism criterion. 

We now use the OWA operator to aggregate the costs of a given schedule $\pi$ under scenarios in ~$\Gamma$.  Let us define
$$\mathrm{OWA}(\pi)={\rm owa}_{\pmb{v}}(f(\pi,S_1),\dots,f(\pi,S_K))=\sum_{i\in [K]} v_i f(\pi,S_{\sigma(i)}),$$
where  $\sigma$ is a permutation of $[K]$ such that $f(\pi,S_{\sigma(1)})\geq \dots\geq f(\pi,S_{\sigma(K)})$. In this paper we examine the following optimization problem:
$$\textsc{Min-Owa}~\mathcal{P}: \min_{\pi\in \Pi} \mathrm{OWA} (\pi).$$
We will also investigate the special cases of the problem, which are listed in Table~\ref{tabsc}.
\begin{table}[ht]
\caption{Special cases of \textsc{Min-Owa}~$\mathcal{P}$.} \label{tabsc}
\begin{tabular}{ll}
 \hline
	  Name of the problem & Weight distribution \\ \hline
		\textsc{Min-Max}~$\mathcal{P}$ & $v_1=1$ and $v_j=0$ for $j=2,\dots,K$ \\
		\textsc{Min-Min}~$\mathcal{P}$ & $v_K=1$ and $v_j=0$ for $j=1,\dots,K-1$ \\
		\textsc{Min-Average}~$\mathcal{P}$ & $v_j=1/K$ for $j\in [K]$ \\
		\textsc{Min-Quant}$(k)$~$\mathcal{P}$ & $v_k=1$ and $v_j=0$ for $j\in [K]\setminus \{k\}$ \\
		\textsc{Min-Median}~$\mathcal{P}$ & $v_{\lfloor K/2 \rfloor +1}=1$ and $v_j=0$ for $j\in [K] \setminus \{\lfloor K/2 \rfloor +1\}$\\
		\textsc{Min-Hurwicz}~$\mathcal{P}$ & $v_1=\alpha$, $v_K=1-\alpha$, $\alpha\in [0,1]$ and $v_j=0$ for $j\in [K]\setminus\{1,K\}$ \\ \hline
\end{tabular}
\end{table} 

Notice that \textsc{Min-Owa}~$\mathcal{P}$ can be consistent with 
a concept of robustness. Namely, the  risk averse decision makers should choose nonincreasing weights, i.e. such that $v_1\geq v_2\geq \dots \geq v_K$. In the extreme case, this leads to the maximum criterion and the \textsc{Min-Max}~$\mathcal{P}$ problem. However, the OWA operator allows us to weaken the maximum criterion by taking more scenarios into account
As we will see in the next sections, the complexity of \textsc{Min-Owa}~$\mathcal{P}$ depends on the properties of the underlying deterministic  problem~$\mathcal{P}$ and the weights $v_1,\dots, v_K$.  One general and easy observation can be made.
Namely, if $\mathcal{P}$ is solvable in $T(n)$ time, then \textsc{Min-Min}~$\mathcal{P}$ is solvable in $O(K\cdot T(n))$ time.
Indeed, in order to solve the \textsc{Min-Min}~$\mathcal{P}$ problem it is enough to compute an optimal schedule $\pi_k$ under each scenario $S_k$, $k\in [K]$, and choose the one which has the minimum value of~$f(\pi_k, S_k)$, $k\in [K]$. For the remaining problems listed in Table~\ref{tabsc} no such general result can be established and their complexity depends on a structure of the deterministic problem~$\mathcal{P}$.

\section{The maximum weighted tardiness cost function} 
\label{sec2}

Let $T_j(\pi,S_i)=[C_j(\pi,S_i)-d_j(S_i)]^+$ be the \emph{tardiness} of job $j$ in $\pi$ under scenario $S_i$, $i\in [K]$. The cost of schedule $\pi$ under $S_i$ is the \emph{maximum weighted tardiness} under $S_i$, i.e. $f(\pi,S_i)=\max_{j\in J} w_j T_j(\pi,S_i)$. The underlying deterministic problem~$\mathcal{P}$  is denoted by 
 $1|prec|\max w_jT_j$ in 
 Graham's notation~\cite{GLLK79}.
  In this section we will also discuss a special case of this problem, denoted by $1||T_{\max}$,  with 
  unit job weights and
  no precedence constraints between the jobs. The deterministic $1|prec|\max w_j T_j$ problem can be solved in $O(n^2)$ time by the well known algorithm designed by Lawler~\cite{LA73}. It follows directly from the Lawler's algorithm that $1||T_{\max}$ can be solved in $O(n\log n)$ time by applying the EDD rule, i.e. by ordering the jobs with respect to  nondecreasing due dates. 
  
 This section contains the following results.  
 We will  consider first the case when $K$ is unbounded ($K$ is a part of the input).
 We will show that  the problems of minimizing the average  cost or median of the costs are then strongly NP-hard and also hard to approximate. On the other hand, we will prove that the problems of minimizing the maximum cost or the Hurwicz criterion are solvable in polynomial time. We will consider next the problem with a constant $K$. It turns out that in this case the general problem of minimizing the OWA criterion can be solved in pseudopolynomial time. Finally, we will propose an approximation algorithm, which can be efficiently applied to some particular weight distributions in the OWA criterion.

\subsection{Hardness of the problem}
\label{sec1_1}

The following theorem characterizes the complexity of the problem:
\begin{thm}
\label{thm1}
		 If the number of scenarios is unbounded, then
		\begin{enumerate}
				\item[(i)] \textsc{Min-Average}~$1||T_{\max}$ is strongly NP-hard and not approximable within $7/6-\epsilon$ for any $\epsilon>0$ unless P=NP,
				\item[(ii)]  \textsc{Min-Median}~$1|| T_{\max}$ is strongly NP-hard and not at all approximable unless P=NP.
		\end{enumerate}
		Furthermore, both assertions remain true even for  jobs with  unit processing times under all scenarios.
\end{thm}
\begin{proof}
	We show a polynomial time approximation preserving reduction from the \textsc{Min $k$-Sat} problem, which is defined as follows.  We are given boolean variables $x_1,\dots,x_n$ and a collection of clauses $C_1,\dots, C_m$, where each clause is a disjunction of at most $k$ literals (variables or their negations). We ask if there is an assignment to the variables which satisfies at most $L<m$ clauses. This problem is strongly NP-hard even for $k=2$ (see~\cite{AZ02, KM94, MR96}) and its optimization (minimization) version is hard to approximate within $7/6-\epsilon$ for any $\epsilon>0$ when $k=3$ (see~\cite{AZ02}).	
	\begin{table}[ht]
	  \centering
	  \caption{The due date scenarios for the formula $(x_1\vee \overline{x}_2 \vee \overline{x}_3)\wedge (\overline{x}_2 \vee \overline{x}_3 \vee x_4) \wedge (\overline{x}_1 \vee  x_2 \vee \overline{x}_4) \wedge (x_1 \vee x_2 \vee x_3) \wedge (x_1 \vee x_3 \vee \overline{x}_4)$.} \label{tab1}
			\begin{tabular}{l|lllll}
										 & $S_1$ & $S_2$ & $S_3$ & $S_4$ & $S_5$ \\ \hline
					$J_{x_1}$ & 1  & 2 & 2 & 1 & 1\\
					$J_{\overline{x}_1}$  & 2 & 2 & 1 & 2 & 2 \\  \hline
					$J_{x_2}$ &  4 & 4 & 3 & 3 & 4\\
					$J_{\overline{x}_2}$ & 3 & 3 & 4 & 4 & 4 \\  \hline
					$J_{x_3}$ & 6 & 6 & 6 & 5 & 5\\
					$J_{\overline{x}_3}$ & 5 & 5 & 6 & 6 & 6 \\  \hline
					$J_{x_4}$ &  8 & 7 & 8 & 8 & 8\\
					$J_{\overline{x}_4}$ & 8 & 8 & 7 & 8 & 7 \\  \hline
			\end{tabular}
		\end{table}
	
	We first consider assertion~(i).
	Given an instance of \textsc{Min 3-Sat}, we construct the corresponding instance of \textsc{Min-Average}~$1||T_{\max}$ in the following way.
	We create two jobs $J_{x_i}$ and $J_{\overline{x}_i}$ for each variable $x_i$, 
	$i\in [n]$.
	 The processing times and weights of all the jobs under all scenarios are equal to~1. 
	 The due dates of $J_{x_i}$ and $J_{\overline{x}_i}$ depend on scenario and will take the value of   
	 either  $2i-1$ or $2i$. Set $K=m$ and 
	 form  
	 $K$ scenario set~$\Gamma$ in the following way. Scenario~$S_k$ corresponds to clause $C_k=(l_1 \vee l_2 \vee l_3)$. For each $q=1,2,3$, if $l_q=x_i$, then the due date of $J_{x_i}$ is $2i-1$ and the due date of $J_{\overline{x}_i}$ is $2i$; if $l_q=\overline{x}_i$, then the due date of $J_{x_i}$ is $2i$ and the due date of $J_{\overline{x}_i}$ is $2i-1$; if neither $x_i$ nor $\overline{x}_i$ appears in $C_k$, then the due dates of $J_{x_i}$ and $J_{\overline{x}_i}$ are set to $2i$.
	  A sample reduction is shown in Table~\ref{tab1}.
	Finally, we fix $v_k=1/m$ for all $k\in [K]$.
	Let us define a subset of the schedules $\Pi'\subseteq \Pi$ such that each schedule $\pi\in \Pi'$ is of the form $\pi=(J_1,J_1',J_2,J_2',\dots,J_n,J_n')$, where $J_i,J_i'\in\{J_{x_i},J_{\overline{x}_i}\}$ for $i\in [n]$. Observe that $\Pi'$ contains exactly $2^n$ schedules and each such a schedule defines an assignment to the variables such that $x_i=0$ if $J_{x_i}$ is processed before $J_{\overline{x}_i}$ and $x_i=1$ otherwise.
	Assume that the answer to \textsc{Min 3-Sat} is yes. So, there is an assignment to the variables which satisfies at most $L$ clauses. Choose schedule~$\pi\in \Pi'$ which corresponds to this assignment. It is easily seen that if clause~$C_k$ is not satisfied, then all jobs in~$\pi$ under~$S_k$ are on-time and the maximum tardiness in~$\pi$ under~$S_k$ is~0.  On the other hand, if clause~$C_k$ is satisfied, then the maximum tardiness of~$\pi$ under~$S_k$ is~1. In consequence
	 $\frac{1}{K}\sum_{k \in [K]} f(\pi,S_k)\leq L/m$.
	Assume now that there is a schedule $\pi$ such that $\frac{1}{K}\sum_{k \in [K]} f(\pi,S_k)\leq L/m$. Notice that $L/m < 1$ by the nonrestrictive assumption that $L<m$. We first show that $\pi$ must belong to $\Pi'$. Suppose that $\pi \notin \Pi'$ and let $J_i$ ($J_i')$ be the last job in $\pi$ which is not placed properly, i.e. $J_i,(J_i')\notin\{J_{x_i},J_{\overline{x}_i}\}$. Then $J_i$ ($J_i'$) is at least one unit late under all scenarios and $\frac{1}{K}\sum_{k \in [K]} f(\pi,S_k)\geq 1$, a contradiction. Since $\pi\in \Pi'$ and all processing times are equal to~1 it follows that $f(\pi,S_k)\in \{0,1\}$ for all $k\in [K]$. Consequently, the maximum tardiness in~$\pi$ is equal to~1 under at most~$L$ scenarios and the assignment corresponding to $\pi$ satisfies at most $L$ clauses.
	The  above 
	reduction is approximation-preserving and the inapproximability result immediately holds.
	
	In order to prove assertion~(ii), it suffices to modify the previous reduction.
	Assume first that  $L<\lfloor m/2 \rfloor$. We then add to scenario set~$\Gamma$ 
	additional $m-2L$ scenarios with the due dates equal to~0 for all the jobs. So the number of scenarios~$K$ is $2m-2L$. 
	We fix $v_{m-L+1}=1$ and $v_k=0$ for the remaining scenarios.
	 Now, the answer to \textsc{Min 3-Sat} is yes, if and only if there is a schedule~$\pi$ 
	 whose maximum tardiness is positive under at most $L+m-2L=m-L$ scenarios. According to the definition of the weights $\mathrm{OWA}(\pi)=0$.
	Assume that $L>\lfloor m/2 \rfloor$. We then add to~$\Gamma$ 
	additional $2L-m$ scenarios with the due dates to~$n$ for all the jobs.  
	The number of scenarios~$K$ is then $2L$. We fix $v_{L+1}=1$ and $v_k=0$ for all the remaining scenarios.
	 Now, the answer to \textsc{Min 3-Sat} is yes, if and only if there is a schedule~$\pi$ whose cost is positive under at most $L$ scenarios. According to the definition of the weights $\mathrm{OWA}(\pi)=0$. We thus can see that it is NP-hard to check whether there is a schedule~$\pi$ such that ${\rm OWA}(\pi)\leq 0$ and the theorem follows.
\end{proof}

The next theorem characterizes the problem complexity when job processing times and due dates are deterministic and only job weights are uncertain.
\begin{thm}
\label{thm2}
		If the number of scenarios is unbounded, then
		\begin{enumerate}
				\item[(i)] \textsc{Min-Average}~$1||\max w_j T_j$ is strongly NP-hard.
				\item[(ii)]  \textsc{Min-Median}~$1||\max w_j T_j$  is strongly NP-hard and not at all approximable  unless P=NP.
		\end{enumerate}
		Furthermore, both assertions are true when all jobs have unit processing times under all scenarios and all job due dates are deterministic.
\end{thm}
\begin{proof}
As in the proof of Theorem~\ref{thm1}, we show a polynomial time reduction from the \textsc{Min 3-Sat} problem. We start by proving assertion (i).  We create two jobs $J_{x_i}$ and $J_{\overline{x}_i}$ for each variable $x_i$. The processing times of these jobs under all scenarios are 1 and their due dates are equal to $2i-1$.
	Now for each clause $C_k=(l_1 \vee l_2 \vee l_3)$ we form the weight scenario $S_k$ as follows:  for each $q=1,2,3$, if $l_q=x_i$, then the weight of $J_{x_i}$ is $1$ and the weight of $J_{\overline{x}_i}$ is $0$; if $l_q=\overline{x}_i$, then the weight of $J_{\overline{x}_i}$ is $1$ and the weight of $J_{x_i}$ is $0$; if neither $x_i$ nor $\overline{x}_i$ appears in $C_k$, then the weights of $J_{x_i}$ and $J_{\overline{x}_i}$ are 0. We also add one additional scenario $S_{m+1}$ under which the weight of each job is equal to $m$.
	We set $K=m+1$ and	
	 fix $v_k=1/(m+1)$ for each $k\in [K]$.  We define the subset of schedules $\Pi'\subseteq \Pi$ as in the proof of Theorem~\ref{thm1}.
	
	We will show that the answer to \textsc{Min 3-Sat} is yes if and only if there is a schedule $\pi$ such that ${\rm OWA}(\pi)\leq (m+L)/(m+1)$. Assume that
	 the answer to \textsc{Min 3-Sat} is yes. Let $\pi\in \Pi'$ be the schedule corresponding to the assignment  which satisfies at most $L$ clauses (see the proof of Theorem~\ref{thm1}). It is easy 
	 to verify that $f(\pi,S_k)=0$ if $C_k$ is not satisfied and $f(\pi,S_k)=1$ if $C_k$ is satisfied. Furthermore, $f(\pi,S_{m+1})=m$. 
	 Hence ${\rm OWA}(\pi)\leq (m+L)/(m+1)$. Assume now that ${\rm OWA}(\pi)\leq (m+L)/(m+1)$. Then $\pi$ must belong to $\Pi'$ since otherwise $f(\pi, S_{m+1})\geq 2m$ and ${\rm OWA}(\pi)\geq 2m/(m+1)$, which contradicts the assumption that $L<m$. It must hold $f(\pi, S_{m+1})=m$ and $f(\pi, S_i)\in \{0,1\}$ for each $i\in [K]$. 
	 Consequently $f(\pi,S_i)=1$ under at most $L$ scenarios, which means that the assignment corresponding to $\pi$ satisfies at most $L$ clauses and 
	 the answer to \textsc{Min 3-Sat} is yes. 
	 
	 The proof of assertion~(ii) is very similar to the corresponding proof in Theorem~\ref{thm1}.
\end{proof}

\subsection{Polynomially and pseudopolynomially solvable cases}

In this section we identify some special cases of the 
\textsc{Min-Owa}~$1|prec|\max w_j T_j$
problem which are polynomially or pseudopolynomially solvable. 

\subsubsection{The maximum criterion}

It has been shown in~\cite{AC08} that  \textsc{Min-Max}~$1|prec|T_{\max}$  is solvable in $O(Kn^2)$ time. In this section, we will show that  more general version of the problem with arbitrary nonnegative job weights, \textsc{Min-Max}~$1|prec|\max w_j T_j$, 
is solvable in $O(Kn^2)$ time as well.
 In the construction of the algorithm,  we will use some ideas from~\cite{K05, VD10}.  Furthermore, the algorithm with some minor modifications will be a basis for solving other special cases of  
 \textsc{Min-Owa}~$1|prec|\max w_j T_j$. 
 In this section the OWA operator is the maximum, so 
  ${\rm OWA}(\pi)=\max_{i\in [K]} f(\pi,S_i)$. By interchanging
   the maximum operators and some easy transformations, 
   we can express the value of $\mathrm{OWA}(\pi)$ as follows:
\begin{equation}
\label{defTF}
\mathrm{OWA}(\pi)  = \max_{j\in J}\max_{i\in [K]} [w_j(S_i)(C_j(\pi,S_i)-d_j(S_i))]^+.
\end{equation}
Fix a nonempty subset of jobs $D\subseteq J$ and define 
\begin{equation}
\label{defFj}
F_j(D)=\max_{i\in [K]} [w_j(S_i) (\sum_{k\in D} p_k(S_i)-d_j(S_i))]^+.\\
\end{equation}
The following proposition  immediately follows from the fact that all job processing times and weights are nonnegative:
\begin{prop}
\label{prop2}
If  $D_2\subseteq D_1$, then for any $j\in J$ it holds $F_j(D_1)\geq F_j(D_2)$.
\end{prop}
Let $pred(\pi,j)$ be the set of jobs containing job $j$ and all the jobs that precede $j$ in $\pi$. 
Since $C_j(\pi,S_i)=\sum_{k\in pred(\pi,j)} p_k(S_i)$, 
the maximum cost of~$\pi$  over $\Gamma$ can be expressed as follows (see~(\ref{defTF}) and~(\ref{defFj})): 
\begin{equation}
\label{defFtmax}
		{\rm OWA}(\pi)=\max_{j\in J} F_j(pred(\pi,j)).
\end{equation}

Consider the algorithm shown in the form of Algorithm~\ref{alg1}.
\begin{algorithm}
    \begin{algorithmic}[1]
        \STATE $D := \{1,\dots,n\}$
        \FORALL {$i\in[K]$} \label{alg1li2}
           \STATE $p(S_i):= \sum_{k\in D} p_k(S_i)$ \label{alg1li3}
        \ENDFOR \label{alg1li4}
        \FOR {$r:= n$ $\bf downto$ $1$}
        \STATE Find $j \in D$, which has no successor in $D$ and has the minimum value of $F_j(D)=\max_{i\in[K]}[w_j(S_i)(p(S_i)-d_j(S_i))]^+$\label{alg1l3}
        \STATE $\pi(r):= j$
        \STATE $D := D\setminus \{j\}$
        \FORALL {$i\in[K]$} \label{alg1li9}
           \STATE $p(S_i):= p(S_i)- p_j(S_i)$
        \ENDFOR \label{alg1li11}
        \ENDFOR
        \STATE \begin{bf} return \end{bf} $\pi$
\end{algorithmic}
\caption{Algorithm for solving \textsc{Min-Max}~$1|prec|\max w_j T_j$.} \label{alg1}
\end{algorithm}
\begin{thm}
\label{thm3}
 Algorithm~\ref{alg1} computes an optimal schedule for \textsc{Min-Max}~$1|prec|\max w_j T_j$ in $O(Kn^2)$ time.
\end{thm}
\begin{proof}
Let $\pi$ be the schedule returned by the algorithm. It is clear that $\pi$ is feasible. Let us renumber the jobs so that $\pi=(1,2,\dots,n)$. Let $\sigma$ be an optimal minmax schedule.  Assume that $\sigma(j)=j$ for $j=k+1,\dots,n$, where $k$ is the smallest position among all the optimal minmax schedules. If $k=0$, then we are done, because $\pi=\sigma$ is optimal. Assume that $k>0$, and so $k\neq \sigma(k)=i$. Let us move the job $k$ just after $i$ in $\sigma$ and denote the resulting schedule as $\sigma'$ 
(see 
Figure~\ref{figtmax}).
 Schedule $\sigma'$ is feasible, because $\pi$ is feasible. 

\begin{figure}[ht]
	\centering
		\includegraphics{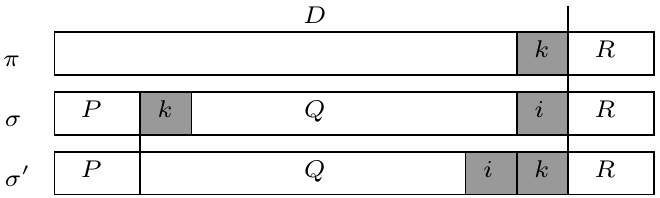}
	\caption{Illustration of the proof of Theorem~\ref{thm1}.}\label{figtmax}
\end{figure}

We need only consider three cases:
\begin{enumerate}
		\item If $j \in P \cup R$, then $pred(\sigma',j)=pred(\sigma,j)$ and $F_j(pred(\sigma',j))=F_j(pred(\sigma,j))$.
		\item If $j \in Q\cup\{i\}$, then $pred(\sigma',j)\subseteq pred(\sigma,j)$ and,  by Proposition~\ref{prop2}, $F_j(pred(\sigma',j))\leq F_j(pred(\sigma,j))$.
		\item If $j=k$, then $F_j(D)\leq F_i(D)$ from the construction of Algorithm~1. 
		Since $pred(\sigma,i)=pred(\sigma',j)=D$, we have $F_j(pred(\sigma',j))\leq F_i(pred(\sigma,i))$.
\end{enumerate}	
From the above three cases and equality~(\ref{defFtmax}), we conclude that 
$${\rm OWA}(\sigma')=\max_{j\in J} F_j(pred(\sigma',j))\leq \max_{j\in J} F_j(pred(\sigma,j))={\rm OWA}(\sigma),$$ 
so $\sigma'$ is also optimal, which contradicts the minimality of $k$.  Computing  $F_j(D)$ for a given $j\in D$ in line~\ref{alg1l3} 
		requires $O(K)$ time
		(note that $p(S_i)$, $i\in [K]$, store the values of $ \sum_{k\in D} p_k(S_i)$
		that have been  computed in lines~\ref{alg1li2}-\ref{alg1li4}
		and they are updated in lines~\ref{alg1li9}-\ref{alg1li11}), 
		and thus line~\ref{alg1l3} can be executed in $O(Kn)$ time. Consequently,
		 the overall running time of the algorithm is $O(Kn^2)$.
\end{proof}

\subsubsection{The Hurwicz criterion}
In this section we explore the problem with the Hurwicz criterion. We will examine the case in
 which $\alpha \in (0,1)$ as the boundary cases with  $\alpha$ equal to 0 (the minimum criterion) or 1 (the maximum criterion) are solvable in $O(Kn^2)$ time. 
\begin{thm}
 \textsc{Min-Hurwicz}~$1|prec|\max w_j T_j$ is solvable in  $O(K^2 n^4 )$ time.
\end{thm}
\begin{proof}
 The Hurwicz criterion can be expressed as follows:
$$\textrm{OWA}(\pi)=\alpha\max_{i\in [K]}f(\pi,S_i) +(1-\alpha)\min_{i\in [K]} f(\pi,S_i).$$
Let us define 
$$H_k(\pi)=\alpha\max_{i\in [K]}f(\pi,S_i) +(1-\alpha) f(\pi,S_k).$$
Hence
\[
\min_{\pi \in \Pi} \textrm{OWA}(\pi)=\min_{k\in [K]} \min_{\pi \in \Pi} H_k(\pi),
\]
and the problem of minimizing the Hurwicz criterion reduces to solving $K$ auxiliary problems consisting in minimizing $H_k(\pi)$ for a fixed $k\in [K]$. 
Let us fix $k\in [K]$ and $t\geq 0$, and  define $\Pi_k(t)=\{\pi\in \Pi\,:\,f(\pi,S_k)\leq t\}\subseteq \Pi$ as the set of feasible schedules whose cost under $S_k$ is at most $t$.
Define 
$$\Psi_k(t)= \min_{\pi \in \Pi_k(t)}\max_{i\in [K]} f(\pi,S_i).$$
Hence
\begin{equation}
\label{defphi}
	\min_{\pi \in \Pi} H_k(\pi)=\min_{t\in [\underline{t}, \overline{t}]} \alpha \Psi_k(t)+(1-\alpha) t,
\end{equation}
where $\underline{t}=\min_{\pi \in \Pi} f(\pi, S_k)$ (for $t<\underline{t}$  it holds $\Pi_k(t)=\emptyset$),
 and $\overline{t}=\min_{\pi \in \Pi}\max_{i\in [K]} f(\pi,S_i)$, which
is due to the fact that $\max_{i\in [K]} f(\pi,S_i)\geq f(\pi,S_k)$.
Computing the value of~$ \Psi_k(t)$ for a given $t\in [\underline{t}, \overline{t}]$ can be done
by a slightly  modified Algorithm~\ref{alg1}. 
It is enough to replace line~\ref{alg1l3} of Algorithm~\ref{alg1} with 
the following line:
\[ 
\ref{alg1l3}': \text{find $j \in D_k(t)$, which has no successor in $D$, and has a minimum value of $F_j(D)$},
\]
 where $D_k(t)=\{j\in D: [w_j(S_k)(p(S_k)-d_j(S_k))]^+\leq t\}$.
The proof of the correctness of the modified algorithm is almost the same as the proof of Theorem~\ref{thm3}.  It is sufficient to define a feasible schedule~$\pi$ as the one satisfying the precedence constraints and the additional constraint $f(\pi,S_k)\leq t$. Hence, if the algorithm returns a feasible schedule, then it must be optimal. The algorithm fails to compute a feasible schedule when $D_k(t)=\emptyset$  in line~\ref{alg1l3}'. In this case, at least one job in $D\neq\emptyset$ must be completed not earlier than $p(S_k)=\sum_{j\in D} p_j(S_k)$ and $f(\pi,S_k)>t$ for all schedules $\pi\in \Pi$, which means that $\Pi_k(t)=\emptyset$. Clearly, the modified algorithm has the same $O(Kn^2)$ running time.

Note that $\Psi_k$ is a nonincreasing step function on $[\underline{t},\infty)$, i.e.
a constant function on subintervals $[\underline{t}_1,\overline{t}_1)\cup[\underline{t}_2,\overline{t}_2)\cup
\cdots\cup [\underline{t}_l,\infty)$, $\overline{t}_{v-1}=\underline{t}_v$, $v=2,\ldots,l$, $\underline{t}_1=\underline{t}$.
Thus,  $ \alpha \Psi_k(t)+(1-\alpha) t$, $\alpha\in(0,1)$, is a piecewise linear function on $[\underline{t},\infty)$,
a linear increasing function on each subinterval
$[\underline{t}_v,\overline{t}_v)$, $v\in [l]$, and attains minimum at one of the points $\underline{t}_1,\dots, \underline{t}_l$. The functions $\Psi_k(t)$ and $\alpha\Psi_k(t)+(1-\alpha)t$ for $k=3$ are depicted  in the example shown in
Figure~\ref{figex}. We have $\underline{t}_1=18$, $\underline{t}_2=26$,  $\underline{t}_3=60$ and the function $\alpha \Psi_3(t)+(1-\alpha)t$ is minimized for $t=26$.  Since $\pi_2=(1,4,2,5,3)$ is an optimal solution to~$\Psi_3(26)$, we conclude that $\pi_2$ minimizes $H_3(\pi)$. 
\begin{figure}[ht]
	\centering
		\includegraphics{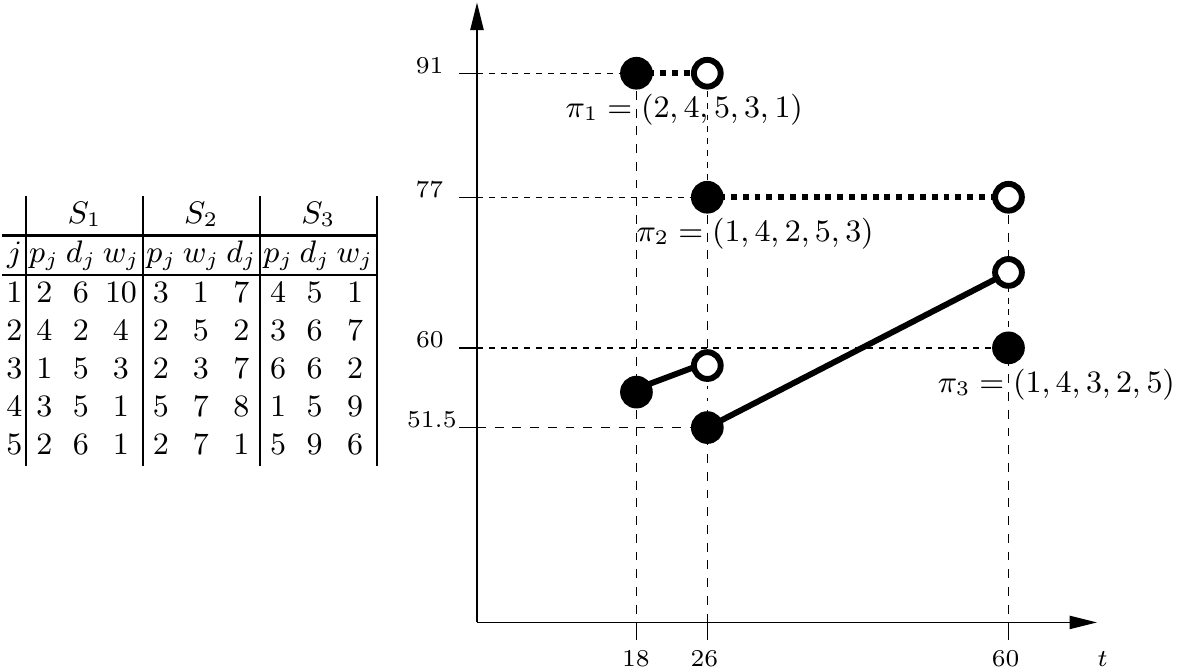}
	\label{figtmax1}
	\caption{The functions $\Psi_3(t)$ (the dotted line) and
	 $0.5\Psi_3(t)+0.5t$, $t\in [18,60]$ (the solid line),  for a sample problem (there are no precedence constraints between the jobs). The function $H_3(\pi)$ is minimized for $\pi_2=(1,4,2,5,3)$ and $H_3(\pi_2)=51.5$.} \label{figex}
\end{figure}

Observe that the value of $t$ minimizing $\alpha \Psi_k(t) + (1-\alpha) t$ can be found in pseudopolynomial time by trying all integers in the interval $[\underline{t},\overline{t}]$.
 We now show how to find the optimal value of $t$ in polynomial time. 
We first compute  $\underline{t}_1=\min_{\pi\in\Pi} f(\pi,S_k)$,  and the value of $\Psi_k(\underline{t}_1)$ by the modified Algorithm~\ref{alg1}.
Let us denote by~$\pi_1$ the resulting optimal schedule, $\pi_1\in \Pi_k(\underline{t}_1)$. In the sample problem shown in Figure~\ref{figex}, $\underline{t}_1=18$, $\pi_1=(2,4,5,3,1)$, and $\Psi_3(\underline{t}_1)=91$. Our goal now is to compute the value of $\underline{t}_2$. 
Choose the iteration of the 
 modified 
Algorithm~\ref{alg1}, in which the position of job $j$ is fixed   in $\pi_1$.  The job  $j$ satisfies the condition stated in line~\ref{alg1l3}'.   We can now compute the smallest value of $t$, $t>\underline{t}_1$,  for which job $j$ violates this condition and must be replaced by some other job in $D_k(t)$. In order to do this it suffices to try all values $t_i=w_i(S_k)[p(S_k)-d_i(S_k)]^+$ for
  $i\in D\setminus\{j\}$ and fix $t^*_j$ as the smallest among them which violates 
  the condition in line~\ref{alg1l3}' (if the condition holds for all $t_i$, then $t^*_j=\infty$). Repeating this procedure for each job we get the set of values $t^*_1,\dots, t^*_n$ and $\underline{t}_2$ is the smallest value among them. Consider again the sample problem presented in Figure~\ref{figex}. When job~1 is placed at position 5 in $\pi_1$, it satisfies the condition  in line~\ref{alg1l3}'  for $t=18$. In fact, it holds $D_3(\underline{t}_1)=\{1\}$. Since $D=\{1,2,3,4,5\}$, we now try the values $t_2=91$, $t_3=26$, $t_4=126$, and $t_5=60$. The condition in line~\ref{alg1l3}' is violated for $t=t_3=26$ as $D_3(26)=\{1,3\}$ and $F_3(D)<F_1(D)$. Hence $t^*_1=26$. In the same way we compute the remaining values $t_2^*,\dots, t^*_5$. It turns out that $t_1^*=26$ is the smallest among them, thus $\underline{t}_2=26$.
  The value of $\underline{t}_3$ can be found in the same way. We compute an optimal
   schedule~$\pi_2$  corresponding to $\Psi_k(\underline{t}_2)$ and repeat the previous procedure. 
  
  Consider the sequence of schedules $\pi_l, \pi_{l-1},\dots,\pi_1$, where $\pi_v$ minimizes $\Psi_k(\underline{t}_v)$.  Schedule $\pi_{v-1}$ can be obtained from $\pi_v$ by moving the position of at least one job in $\pi_v$, say $j$, whose current position becomes infeasible as $t$ decreases, to the left. Furthermore the position of $j$ cannot increase in all the subsequent schedules $\pi_{v-2}, \dots, \pi_1$, because the function $f(\pi,S_k)$ is nondecreasing (if $j$ cannot be placed at $i$th position, then it also cannot be placed at positions $i+1,\dots,n$). Hence, if $\pi_l$ is the last schedule, then the position of job $\pi_l(i)$ can be decreased at most $i-1$ times which implies $l=O(n^2)$. Hence problem~(\ref{defphi}) can be solved in $O(Kn^4)$ time and  \textsc{Min-Hurwicz}~$1|prec|\max w_j T_j$ is solvable in $O(K^2 n^4 )$ time.  
\end{proof}

\subsubsection{The $k$th largest cost criterion}

In this section we investigate the \textsc{Min-Quant($k$)}~$1|prec|\max w_j T_j$ problem. Thus our goal is to minimize the $k$th largest schedule cost. It is clear that this problem is polynomially solvable when $k=1$ or $k=K$. It is, however, strongly NP-hard and not at all approximable when $k$ is a function of $K$, in particular, when the median of the costs is minimized (see Theorem~\ref{thm1}). We now explore the case when $k$ is constant. 
\begin{thm}
\label{thm4a}
		 \textsc{Min-Quant($k$)}~$1|prec|\max w_j T_j$  is solvable in $O\left(\binom{K}{k-1} (K-k+1) n^2 \right)$ time, which is polynomial when $k$ is constant.
\end{thm}
\begin{proof}
		The algorithm works as follows. We enumerate all the subsets of scenarios of size $k-1$. For each such a subset, say $C$, we solve \textsc{Min-Max}~$1|prec|\max w_j T_j$ for the scenario set $\Gamma\setminus C$, using Algorithm~\ref{alg1}, obtaining a schedule $\pi_C$. Among the  schedules computed we return~$\pi_C$ for which the maximum cost over $\Gamma\setminus C$ is minimal. It is straightforward to verify that this schedule must be optimal. The number of subsets which have to be enumerated is $\binom{K}{k-1}$. For each such a subset we solve 
		 \textsc{Min-Max}~$1|prec|\max w_j T_j$ with scenarios set~$\Gamma \setminus C$, which requires $O((K-k+1) n^2)$ time and the theorem follows.
\end{proof}
The algorithm suggested in the proof of Theorem~\ref{thm4a} is efficient when $k$ is close to~1 or close to $K$. When $k$ is a function of $K$, then this running times becomes exponential and may be prohibitive in practice. In Section~\ref{sappal1}, we will use this algorithm to construct an approximation algorithm for the general  \textsc{Min-Owa}~$1|prec|\max w_j T_j$ problem.

\subsubsection{The OWA criterion - the bounded case}
\label{sec1_2}

In Section~\ref{sec1_1}, we have shown that for the unbounded case 
\textsc{Min-Owa}~$1|prec|\max w_j T_j$  is strongly NP-hard and not at all approximable unless P=NP. In this section we investigate the case when $K$ is constant.
Without loss of generality we can assume
that all the parameters are nonnegative integers.
Let $f_{\max}$ be an upper bound on the maximum weighted tardiness of any job under any scenario. 
By Proposition~\ref{prop2} and equality~(\ref{defFtmax}) we can fix $f_{\max}=\max_{j\in J} F_j(J)$.

\begin{thm}
		\textsc{Min-Owa}~$1|prec|\max w_j T_j$ is solvable in $O(f_{\max}^K Kn^2)$ time, which is pseudopolynomial if $K$ is constant.
		\label{thm4}
\end{thm}
\begin{proof}
	Let $\pmb{t}=(t_1,\dots,t_K)$ be a vector of nonnegative integers.
	Let $\Pi(\pmb{t})\subseteq \Pi$ be a subset of the set of feasible schedules such that $\pi \in \Pi(\pmb{t})$ if $f(\pi,S_i)\leq t_i$ for all $i\in [K]$, i.e. the maximum weighted tardiness in $\pi$ under $S_i$ does not exceed $t_i$. Consider the following auxiliary problem. Given a vector $\pmb{t}$, check if $\Pi(\pmb{t})$ is not empty and if so, return any schedule $\pi_{\pmb{t}}\in \Pi(\pmb{t})$. We now show that this auxiliary problem can be solved in polynomial time. Given $\pmb{t}$,  we first form scenario set $\Gamma'$ by specifying the following parameters  for each $S_i \in \Gamma$ and $j\in J$:
	\begin{itemize}
	\item $p_j(S_i')=p_j(S_i)$, 
	\item $ \displaystyle d_j(S_i')=\max\{C\geq 0\,:\,w_j(S_i)(C-d_j(S_i))\leq t_i\}=t_i/w_j(S_i)+d_j(S_i)$,
	\item $w_j(S_i')=1$.
	\end{itemize}
	The scenario set $\Gamma'$ can be determined in $O(Kn)$ time.
We solve \textsc{Min-Max}~$1|prec|\max w_j T_j$ with  the scenario set $\Gamma'$ by 
Algorithm~\ref{alg1} obtaining schedule $\pi$. If the maximum cost of $\pi$ over $\Gamma'$ is 0, then $\pi_{\pmb{t}}=\pi$; otherwise $\Pi(\pmb{t})$ is empty. Since 
\textsc{Min-Max}~$1|prec|\max w_j T_j$
is solvable in $O(Kn^2)$ time, the auxiliary problem is solvable in $O(Kn^2)$ time as well.
	We now show that there exists a vector $\pmb{t}^*=(t_1^*,\dots, t^*_K)$, where $t_i^*\in \{0,\dots,f_{\max}\}$, $i\in [K]$, such that each $\pi_{\pmb{t}^*}\in \Pi(\pmb{t}^*)$ minimizes
	 $\mathrm{OWA(\pi)}$. Let $\pi^*$ be an optimal schedule and let $\pmb{t^*}=(t^*_1,\dots,t^*_K)$ be a vector such that $t^*_i=f(\pi^*,S_i)$ for $i\in [K]$. Clearly, $t^*_i\in \{0,\dots,f_{\max}\}$ for each $i\in [K]$ and $\pi^* \in \Pi(\pmb{t^*})$. By the definition of $\pmb{t^*}$, it holds $\mathrm{owa}_{\pmb{v}}(\pmb{t^*})=\mathrm{OWA}(\pi^*)$. For any $\pi\in \Pi(\pmb{t^*})$ it holds $f(\pi,S_i)\leq t^*_i= f(\pi^*,S_i)$, $i\in [K]$. From the monotonicity of OWA we conclude that each $\pi \in \Pi(\pmb{t^*})$ must be optimal.
	The algorithm enumerates all possible vectors $\pmb{t}$ and computes $\pi_{\pmb{t}}\in \Pi(\pmb{t})$ if $\Pi(\pmb{t})$ is nonempty. A schedule~$\pi_{\pmb{t}}$ with the minimum value 
	of $\mathrm{owa}_{\pmb{v}}(\pmb{t})$ is returned. The number of vectors $\pmb{t}$ which must be enumerated is at most $f_{\max}^K$. Hence the problem is solvable in pseudopolynomial time provided that $K$ is constant and the running time of the algorithm is $O(f_{\max}^K Kn^2)$.
\end{proof}

\subsection{Approximation algorithm}
\label{sappal1}

When $K$ is a part of the input, i.e. in the unbounded case, then the exact algorithm proposed in Section~\ref{sec1_2} may be inefficient. Notice, that due to Theorem~\ref{thm1}, no efficient approximation algorithm can exist for 
 \textsc{Min-Owa}~$1|prec|\max w_j T_j$ in this case
 unless P=NP. We now prove the following result, which can be used to obtain an approximate solution in some special cases of the weight distributions in the OWA operator.
\begin{thm}
\label{thm7}
		Suppose that $v_1=\dots =v_{k-1}=0$ and $v_k>0$, $k\in [K]$. 
		Let $\hat{\pi}$ be an optimal solution to the \textsc{Min-Quant}$(k)$~$1|prec|\max w_j T_j$ problem. Then for each $\pi\in \Pi$, it holds ${\rm OWA}(\hat{\pi})\leq (1/v_k){\rm OWA}(\pi)$ and the bound is tight.
\end{thm}
\begin{proof}
Let $\sigma$ be a sequence of $[K]$ such that $f(\hat{\pi},S_{\sigma(1)})\geq \dots \geq f(\hat{\pi},S_{\sigma(K)})$ and $\rho$ be a sequence of $[K]$ such that $f(\pi,S_{\rho(1)})\geq \dots \geq f(\pi,S_{\rho(K)})$. It holds:
$${\rm OWA}(\hat{\pi})=\sum_{i=k}^K v_i f(\hat{\pi},S_{\sigma(i)})\leq f(\hat{\pi},S_{\sigma(k)}).$$
From the definition of $\hat{\pi}$ and the assumption that $v_k>0$ we get
$$f(\hat{\pi},S_{\sigma(k)})\leq f(\pi,S_{\rho(k)})\leq \frac{1}{v_k}\sum_{i \in [K]} v_i f(\pi,S_{\rho(i)})=\frac{1}{v_k}{\rm OWA}(\pi).$$
Hence ${\rm OWA}(\hat{\pi})\leq (1/v_k){\rm OWA}(\pi)$.
To see that the bound is tight consider 
an instance of
the problem with $K$ scenarios and $2K$ jobs. The job processing times and weights are equal to~1 under all scenarios. The job due dates are shown in Table~\ref{tab2a}. We fix $v_i=(1/K)$ for each $i\in [K]$.

\begin{table}[ht]
\centering
\caption{An example of due date scenario set for which the approximation algorithm achieves a ratio of $1/v_k$.} \label{tab2a}
\begin{tabular}{l|ccccccccccc}
   & $S_1$ & $S_2$ & $S_3$ & $\dots$ & $S_K$ \\ \hline
$J_1$ & 1 & 1 & 1 & $\dots$ & 1 \\
$J_2$ & 2 & 2 & 2 & $\dots$ & 1 \\
$J_3$ & 3 & 3 & 3 & $\dots$ & 3\\
$J_4$ & 4 & 4 & 4 & $\dots$ & 3\\
$\vdots$ & $\vdots$ & $\vdots$ & $\vdots$ & $\vdots$ & $\vdots$   \\
$J_{2K-1}$ & $2K-1$ & $2K-1$ & $2K-1$ & $\dots$ & $2K-1$\\
$J_{2K}$ & $2K$ & $2K$ & $2K$ & $\dots$ & $2K-1$
\end{tabular}
\end{table}

Since $v_1>0$, we solve \textsc{Min-Max}~$1|prec|\max w_j T_j$. As a result we can obtain the schedule $\pi=(J_2,J_1,J_4,J_3,\dots,J_{2K},J_{2K-1})$ whose average cost over all scenarios is $1$. But the average cost of the optimal schedule $\pi^*=(J_1,J_2,J_3,J_4,\dots,J_{2K-1},J_{2K})$ is $1/K$.
\end{proof}

We now show several consequences of Theorem~\ref{thm7}. Observe first that if $v_1>0$, then we can use Algorithm~\ref{alg1} to obtain the approximate schedule in polynomial time.
\begin{cor}
If  $v_1>0$, then \textsc{Min-Owa}~$1|prec|\max w_j T_j$ is approximable within $1/v_1$.
\end{cor}

Consider now the case of nondecreasing weights, i.e. $v_1\geq v_2\geq \dots \geq v_K$. Recall that nondecreasing weights are used when the idea of robust optimization is adopted. Namely, larger weights are assigned to larger schedule costs. Since in this case the inequality $v_1\geq 1/K$ must hold, 
we get the following result:
\begin{cor}
If the weights are nonincreasing, then \textsc{Min-Owa}~$1|prec|\max w_j T_j$  is approximable within $1/v_1\leq K$.
\end{cor}
Finally, the following corollary is an immediate consequence of the previous corollary:
\begin{cor}
 \textsc{Min-Average}~$1|prec|\max w_j T_j$ is approximable within~$K$.
\end{cor}

\section{The weighted sum of completion times cost function}
\label{sec3}

Let the cost of schedule~$\pi$ under scenario~$S_i$ be the weighted sum of completion times in $S_i$, i.e.
 $f(\pi,S_i)=\sum_{j\in J} w_j(S_i) C_j(\pi,S_i)$. Using the Graham's notation, the deterministic version of the problem is denoted 
 by $1|prec|\sum w_j C_j$. 
 We will also examine the special cases of this problem with no precedence constraints between the jobs, i.e. $1||\sum w_j C_j$ and all job weights equal to~1, i.e.  $1||\sum C_j$. It is well known that $1|prec|\sum C_j$ is strongly NP-hard for arbitrary precedence constraints~\cite{LK78}. It is, however, polynomially solvable for some special cases of the precedence constraints such as in-tree, out-tree or sp-graph (see, e.g.~\cite{B07}). If there are no precedence constraints between the jobs, then an optimal schedule can be obtained by ordering the jobs with respect to nondecreasing ratios $p_j/w_j$, which reduces to the SPT rule when all job weights are equal to~1. 
 
 In this section, we will show that if the number of scenarios is a part of the input, then \textsc{Min-Owa}~$1||\sum w_j C_j$ is strongly NP-hard and not at all approximable. This is the case when the weights in the OWA criterion are nondecreasing, or OWA is the median. We then propose several approximation algorithms which will be valid for nonincreasing weights and the Hurwicz criterion.
 
\subsection{Hardness of the problem}

 The \textsc{Min-Max}~$1||\sum w_j C_j$ and \textsc{Min-Max}~$1||\sum C_j$ problems have
  been recently investigated in literature, 
 and the following results have been established:
\begin{thm}[\cite{YY02}]
\label{thmc1}
		\textsc{Min-Max}~$~1||\sum C_j$ is NP-hard  even for $K=2$.
\end{thm}
\begin{thm}[\cite{MNO13}]
\label{thmc2}
		If the number of scenarios is unbounded, then 
		\begin{enumerate}
					\item[(i)] \textsc{Min-max}~$1||\sum w_j C_j$ 
					is strongly NP-hard and not approximable within $O(\log^{1-\varepsilon}n)$ for any $\varepsilon>0$ unless the problems in NP have quasi-polynomial time algorithms.
					\item[(ii)] \textsc{Min-max}~$1||\sum C_j$ and
					 \textsc{Min-max}~$1|p_j=1|\sum w_j C_j$
					 are strongly NP-hard and not approximable within $6/5-\varepsilon$ for any $\varepsilon>0$ unless P=NP.
		\end{enumerate}
\end{thm}
We now show that the general case is much more complex.
\begin{thm} 
\label{thmcc}
	If the number of scenarios is unbounded, then	\textsc{Min-Owa}~$1||\sum w_j C_j$ is strongly NP-hard and not at all approximable
		unless P=NP.
		\end{thm}
\begin{proof}
We show a polynomial time reduction from the \textsc{Min 2-Sat} problem which is known to be strongly NP-hard 
(see the proof of Theorem~\ref{thm1}).
	Given an instance of \textsc{Min 2-Sat}, we construct 
	the corresponding instance of \textsc{Min-Owa}~$1||\sum w_j C_j$  in the following way.
	We associate 
two jobs $J_{x_i}$ and $J_{\overline{x}_i}$ with each variable~$x_i$, 
$i\in [n]$. 
 We then set $K=m$ and form scenario set~$\Gamma$
 in the following way. 
 Scenario $S_k$ corresponds to clause $C_k=(l_1 \vee l_2)$. For $q=1,2$, if $l_q=x_i$, then the processing time of $J_{x_i}$ is $0$, the weight of $J_{x_i}$ is~1,
 the processing time of $J_{\overline{x}_i}$ is $1$, and the weight of $J_{\overline{x}_i}$ is~0;
 if $l_q=\overline{x}_i$, then the  processing time of $J_{x_i}$ is $1$, the weight of $J_{x_i}$ is 0,
 the processing time of $J_{\overline{x}_i}$ is $0$ and the weight of $J_{\overline{x}_i}$ is~1.
 If neither $x_i$ nor $\overline{x}_i$ appears in $C_k$, then both processing times and weights of $J_{x_i}$ and 
$J_{\overline{x}_i}$ are set to 0. 	
	We complete the reduction by fixing $v_1=v_2=\dots =v_L=0$ and $v_{L+1}=\dots v_K=1/(m-L)$.  A sample reduction is presented in Table~\ref{tab2}.
	\begin{table}[ht]
	  \centering
	  \caption{Processing times and weights $(p_j(S_i), w_j(S_i))$ corresponding to 
	the formula $(x_1\vee \overline{x}_2)\wedge (\overline{x}_2 \vee \overline{x}_3) \wedge 
	(\overline{x}_1 \vee \overline{x}_4) \wedge (x_1 \vee x_3) \wedge (x_1 \vee \overline{x}_4)$.} \label{tab2}
			\begin{tabular}{l|lllll}
				                         & $S_1$ & $S_2$ & $S_3$ & $S_4$ & $S_5$ \\ \hline
					$J_{x_1}$ & $(0,1)$  & $(0,0)$ & $(1,0)$ & $(0,1)$ & $(0,1)$\\
					$J_{\overline{x}_1}$  & $(1,0)$ & $(0,0)$ & $(0,1)$ & $(1,0)$ & $(1,0)$ \\  \hline
					$J_{x_2}$ &  $(1,0)$ & $(1,0)$ & $(0,0)$ & $(0,0)$ & $(0,0)$\\
					$J_{\overline{x}_2}$ & $(0,1)$ & $(0,1)$ & (0,0) & $(0,0)$ & $(0,0)$ \\  \hline
					$J_{x_3}$ & (0,0) & $(1,0)$ & (0,0) & $(0,1)$ & $(0,0)$\\
					$J_{\overline{x}_3}$ & (0,0) & $(0,1)$ & $(0,0)$ & $(1,0)$ & $(0,0)$ \\  \hline
					$J_{x_4}$ &  (0,0) & $(0,0)$ & $(1,0)$ & $(0,0)$ & $(1,0)$\\
					$J_{\overline{x}_4}$ & (0,0) & $(0,0)$ & $(0,1)$ & $(0,0)$ & $(0,1)$ \\  \hline
			\end{tabular}
	\end{table}

	 We now show that there is an assignment to the variables which satisfies at most~$L$ clauses 
	 if and only if there is a schedule~$\pi$ such that $\mathrm{OWA}(\pi)=0$.
	  Assume that there is an assignment  $x_i$, $i\in [n]$, that satisfies at most~$L$ clauses. 
	  According to this assignment we build a schedule~$\pi$ as follows. 
	  We first process $n$ jobs $J_{z_i}$, $z_i\in\{x_i, \overline{x}_i\}$, 
	   which correspond to false literals~$z_i$, $i\in[n]$, in any order and 
	   then the rest $n$ jobs that
	    correspond to true literals~$z_i$, $i\in[n]$,
	    in any order.
	     Choose a clause $C_k=(l_1 \vee l_2)$ which is not satisfied. It is easy to check that the cost of 
	     the schedule~$\pi$ under scenario $S_k$ is~0. 
	     Consequently, there are at most~$L$ scenarios under which the cost of~$\pi$ is 
	     positive and, according to the definition of the weights in the OWA operator,
	      we get $\mathrm{OWA}(\pi)=0$.
	       Suppose now that there is a schedule~$\pi$ such that
	        $\mathrm{OWA}(\pi)=0$.
	         We construct an assignment to the variables by setting  $x_i=0$ if $J_{x_i}$ appears before 
	         $J_{\overline{x}_i}$ in $\pi$ and $x_i=1$ otherwise. Since $\mathrm{OWA}(\pi)=0$,
	          the cost of~$\pi$ must be~0 under at least $m-L$ scenarios. If the cost of $\pi$ is~0 under
	           scenario~$S_k$ corresponding to the clause $C_k$, then the assignment does not satisfy $C_k$. Hence, 
	           there is at least $m-L$ clauses that are not satisfied and, consequently, at most $L$ satisfiable clauses. 	
\end{proof}
\begin{cor}
		\textsc{Min-Median}~$1||\sum w_j C_j$ is strongly NP-hard and not at all approximable unless P=NP. 
\end{cor}
\begin{proof}
	The proof is similar to the proof of Theorem~\ref{thm1} and consists in adding some additional scenarios to 
	an instance of problem constructed in Theorem~\ref{thmcc}.
\end{proof}

\subsection{Approximation algorithms}

In this section we show several approximation algorithms for \textsc{Min-Owa}~$1|prec|\sum w_j C_j$.
We will explore the case in which the weights in the OWA criterion are
 nonincreasing, i.e. $v_1\geq v_2 \geq \dots \geq v_K$. We will then apply the obtained results to the Hurwicz criterion. Observe, that the case with nondecreasing weights, i.e. $v_1\leq v_2\leq \dots\leq v_K$, is not at all approximable (see the proof of Theorem~\ref{thmcc}).
 We first recall  a well known property
 (see, e.g.~\cite{MNO13}) which states
  that each problem with uncertain processing times and deterministic weights can be transformed into an equivalent problem with uncertain weights and deterministic processing times (and vice versa). This transformation is cost preserving and works as follows. Under each scenario $S_i$, $i\in [K]$, we invert the role of processing times and weights obtaining scenario $S'_i$. The new scenario set $\Gamma'$ contains scenario $S'_i$ for each $i\in [K]$. We also invert the precedence constraints, i.e. if $i\rightarrow j$ in the original problem, then $j\rightarrow i$ in the new one. It can be easily shown that the cost of schedule $\pi$ under $S$ is equal to the cost of the inverted schedule $\pi'=(\pi(n),\dots,\pi(1))$ under $S'$. Consequently $\mathrm{OWA}(\pi)$ under $\Gamma$ equals $\mathrm{OWA}(\pi')$ under $\Gamma'$. 
  Notice that if the processing times are deterministic in the original problem, then the weights become deterministic in the new one (and vice versa).
  
Let $w_{\max}, w_{\min}, p_{\max}, p_{\min}$ be the largest (smallest) weight (processing time) in the input instance. We first consider the case then both processing times an weights can be uncertain. We prove the following result:
\begin{thm}
	If $v_1\geq v_2\geq \dots \geq v_K$ and the deterministic $1|prec|\sum w_j C_j$ problem is polynomially solvable,  then \textsc{Min-Owa}~$1|prec|\sum w_j C_j$ is approximable within $K\cdot\min\{\frac{w_{\max}}{w_{\min}},\frac{p_{\max}}{p_{\min}}\}$.
	\end{thm}
	\begin{proof}
	Let
	$\hat{p}_j=\sum_{i \in [K]} p_j(S_i)$, $\hat{w}_j={\rm owa}_{\pmb{v}}(w_j(S_1),\dots,w_j(S_K))$, $\hat{C}_j(\pi)=\sum_{i\in [K]} C_j(\pi,S_i)$, and   $\hat{f}(\pi)=\sum_{j\in J} \hat{w}_j\hat{C}_j(\pi)$. Let $\hat{\pi}\in \Pi$ minimize $\hat{f}(\pi)$. Of course, $\hat{\pi}$ can be computed in polynomial time provided that the deterministic counterpart of the problem is polynomially solvable. 
	Let $\sigma$ be a sequence of $[K]$ such that $f(\hat{\pi},S_{\sigma(1)})\geq \dots\geq f(\hat{\pi},S_{\sigma(K)})$.
	It holds	
	\begin{equation}
	\label{cappr0}
	\begin{array}{ll}
	{\rm OWA}(\hat{\pi})=  & \displaystyle \sum_{i\in [K]} v_i\sum_{j\in J} w_j(S_{\sigma(i)})C_j(\hat{\pi},S_{\sigma(i)})\leq\sum_{j\in J} \sum_{i\in [K]} v_iw_j(S_{\sigma(i)})\hat{C}_j(\hat{\pi})= \\
	& \displaystyle=\sum_{j\in J}\hat{C}_j(\hat{\pi}) \sum_{i\in [K]} v_iw_j(S_{\sigma(i)})\leq \sum_{j \in J} \hat{w}_j \hat{C}_j(\hat{\pi})=\hat{f}(\hat{\pi}),
	\end{array}
	\end{equation}
where the inequality $\hat{w}_j\geq \sum_{i\in [K]} v_i w_j(S_{\sigma(i)})$ follows from the assumption that $v_1\geq v_2\geq \dots \geq v_K$.	We also get for any $\pi \in \Pi$
	\begin{equation}
	\label{cappr1}
	\begin{array}{ll}
	\displaystyle \hat{f}(\hat{\pi})\leq &  \displaystyle\hat{f}(\pi)=\sum_{j \in J} \hat{w}_j \hat{C}_j(\pi)=\sum_{j \in J} \hat{w}_j \sum_{i\in [K]} C_j(\pi,S_i)\leq \frac{w_{\max}}{w_{\min}}\sum_{j\in J}\sum_{i\in[K]}w_j(S_i)C_j(\pi,S_i)= \\
	& \displaystyle=\frac{w_{\max}}{w_{\min}}\sum_{i\in[K]}\sum_{j\in J}w_j(S_i)C_j(\pi,S_i),
	\end{array}
	\end{equation}
	where the second inequality follows from the fact that $\hat{w_j}\leq w_{\max}\leq (w_{\max}/w_{\min})w_j(S_i)$ for each $j\in J$, $i\in [K]$.
	Again, from the assumption that  $v_1\geq v_2\geq \dots \geq v_K$ we have
	\begin{equation}
	\label{cappr2}
	(1/K)\sum_{i\in[K]}\sum_{j\in J}w_j(S_i)C_j(\pi,S_i)\leq {\rm OWA}(\pi).
	\end{equation}
	From~(\ref{cappr0}), (\ref{cappr1}) and (\ref{cappr2}) we get
	${\rm OWA}(\hat{\pi})\leq K\cdot \frac{w_{\max}}{w_{\min}}{\rm OWA}(\pi).$
	Since the role of job processing times and weights can be inverted we also get ${\rm OWA}(\hat{\pi})\leq K\cdot \frac{p_{\max}}{p_{\min}}{\rm OWA}(\pi)$ and the theorem follows.
	\end{proof}

In~\cite{MNO13} a 2-approximation algorithm for  \textsc{Min-Max}~$1|prec|\sum w_j C_j$  has been recently proposed, provided that either job processing times or job weights are deterministic
(they do not vary among scenarios). In this section we will show that  this algorithm can be extended to \textsc{Min-Owa}~$1|prec|\sum w_jC_j$ under the additional assumption that the weights in the OWA operator are nonincreasing, i.e. $v_1\geq v_2\geq \dots \geq v_K$.
The idea of the approximation algorithm is to design a mixed integer programming formulation for the problem, solve its linear relaxation and construct an approximate schedule based on the optimal solution to this relaxation.

Assume now that job processing times are deterministic and equal to $p_j$ under each scenario $S_i$, $i\in [K]$.  Let $\delta_{ij}\in \{0,1\}$, $i,j\in [n]$, be binary variables such that $\delta_{ij}=1$ if job $i$ is processed before job $j$ in
 a schedule constructed. The vectors of all feasible job completion times $(C_1,\dots, C_n)$ can be described by the following system of constraints~\cite{POT80}:
	\begin{equation}
	\label{cCC2}
		 \begin {array}{llll}
				 VC: & C_j=p_j+\sum_{i\in J\setminus\{j\}} \delta_{ij} p_i & j\in J\\
				 &\delta_{ij}+\delta_{ji}=1 & i,j\in J, i\neq j \\
				 &\delta_{ij}+\delta_{jk}+\delta_{ki} \geq 1 & i,j,k \in J\\
				 &\delta_{ij}=1 & i\rightarrow j\\
				 &\delta_{ij}\in \{0,1\}& i,j \in J
		\end{array}
\end{equation}
Let us denote by $VC'$ the relaxation of $VC$, in which the constraints $\delta_{ij}\in \{0,1\}$ are replaced with $0\leq \delta_{ij}\leq 1$. 	It has been proved in~\cite{SH96a, SH96b} (see also~\cite{HA97}) that each vector $(C_1,\dots, C_n)$ that satisfies $VC'$ also satisfies the following inequalities:
	\begin{equation}
		\label{Schin}
			\sum_{j\in I} p_jC_j\geq \frac{1}{2}\left((\sum_{j\in I} p_j)^2+\sum_{j\in I} p_j^2\right) \text{ for all } I \subseteq J
	\end{equation}

In order to build a MIP formulation for the problem, we will use the idea of a deviation model introduced in~\cite{OS03}.
Let $\sigma$ be a permutation of $[K]$ such that $f(\pi,S_{\sigma(1)})\geq \dots \geq f(\pi,S_{\sigma(K)})$ and
 let $\theta_k(\pi)=\sum_{i=1}^k f(\pi,S_{\sigma(i)})$ be the cumulative cost of schedule $\pi$. Define $v'_i=v_i-v_{i+1}$ for $i=1,\dots,K-1$ and $v'_K=v_K$. An easy verification shows that
\begin{equation}
\label{mipc00}
		{\rm OWA}(\pi)=\sum_{k=1}^K v'_k \theta_k(\pi).
\end{equation}

\begin{lem}
 	Given $\pi$, the value of $\theta_k(\pi)$ can be obtained by solving the following linear programming problem:
	\begin{equation}
	\label{mipc0}
		\begin{array}{lllll}
				\min & \sum_{i=1}^K  u_{i} - (K-k) r \\
				& r\leq u_{i} & i\in [K] \\
				& u_{i} \geq f(\pi,S_i) & i\in[K] \\
				& u_i \geq 0 & i\in [K]\\
				& r\geq 0
		\end{array}
 	\end{equation}
\end{lem}
\begin{proof}
	Consider the following linear programming problem:
	\begin{equation}
	\label{mipc1}
		\begin{array}{lllll}
				\max & \sum_{i=1}^K  \beta_i f(\pi,S_i) \\
				& \alpha_i + \beta_i \leq 1 & i\in [K] \\
				& \sum_{i=1}^K \alpha_i \geq (K-k)\\
				& \alpha_i, \beta_i \geq 0 & i\in [K]
		\end{array}
 	\end{equation}	
It is easy to see that an optimal solution to~(\ref{mipc1}) can be obtained by setting $\beta_{\sigma(i)}=1$ and $\alpha_{\sigma(i)}=0$ for $i=1\dots k$, $\beta_{\sigma(i)}=0$ and $\alpha_{\sigma(i)}=1$ for $i=k+1,\dots K$,
where $\sigma$ is such that $f(\pi,S_{\sigma(1)})\geq \dots \geq f(\pi,S_{\sigma(K)})$.
This gives us the maximum objective function value equal to $\theta_k(\pi)$. To complete the proof it is enough to observe that~(\ref{mipc0}) is the dual linear program to~(\ref{mipc1}).
\end{proof}
If $v_1\geq v_2\geq\dots\geq v_K$, then $v'_i\geq 0$ and~(\ref{cCC2}), (\ref{mipc00}),  (\ref{mipc0}) lead to the following mixed integer programming formulation for the problem:
\begin{equation}
	\label{mipc2}
		\begin{array}{lllll}
				\min & \sum_{k=1}^K v'_k (\sum_{i=1}^K  u_{ik} - (K-k) r_k) \\
				 & \text{Constraints } VC \\		
				 & r_k\leq u_{ik} & i,k\in[K] \\
				& u_{ik} \geq \sum_{j\in J} C_j w_j(S_i) & i,k\in[K] \\
				& u_{ik} \geq 0 & i,k\in[K]\\
				& r_k\geq 0 & k\in[K] \\
		\end{array}
 	\end{equation}
	
In order to construct the approximation algorithm we will also need the following easy observation:
	
\begin{obs}
\label{lemowa1}
		Let $(f_1,\dots, f_K)$ and $(g_1,\dots,g_K)$ be two 
		nonnegative real vectors
		such that $f_i\leq \gamma g_i$ for some constant $\gamma>0$. Then,
		 $\mathrm{owa}_{\pmb{v}}(f_1,\dots, f_k)\leq \gamma
		  \mathrm{owa}_{\pmb{v}}(g_1,\dots, g_K)$ for each $\pmb{v}$.
\end{obs}
\begin{proof}
 From the monotonicity of the OWA operator and the assumption $\gamma>0$,  it follows that ${\rm owa}_{\pmb{v}}(f_1,\dots,f_K)\leq {\rm owa}_{\pmb{v}}(\gamma g_1,\dots,\gamma g_K)=\gamma {\rm owa}_{\pmb{v}}(g_1,\dots,g_K)$.
\end{proof}

The approximation algorithm works as follows. We first solve the linear relaxation of~(\ref{mipc2}) in which $VC$ is replaced with $VC'$ . Clearly, this relaxation can be solved in polynomial time. Let $(C_1^*, \dots, C_n^*)$ be the relaxed optimal job completion times and $z^*$ be the optimal value of the relaxation. It holds $z^*={\rm owa}_{\pmb{v}}(\sum_{j\in J} C^*_j w_j(S_1), \dots, \sum_{j\in J} C^*_j w_j(S_K))$.
We now relabel the jobs so that $C^*_1\leq C^*_2\leq \dots\ C_n^*$ and form schedule $\pi=(1,2,\dots,n)$.
Since the vector $(C_j^*)$ satisfies $VC'$ it must also satisfy~(\ref{Schin}).  Hence, by setting $I=\{1,\dots,j\}$, we get
	$$\sum_{i=1}^j p_iC^*_i\geq \frac{1}{2}\left((\sum_{i=1}^j p_i)^2+\sum_{i=1}^j p_i^2\right)\geq \frac{1}{2} (\sum_{i=1}^j p_i)^2.
	$$
	Since $C^*_j\geq C_i^*$ for each $i\in \{1\dots j\}$, we get
	$C^*_j\sum_{i=1}^j p_i\geq \sum_{i=1}^j p_iC^*_i \geq \frac{1}{2}(\sum_{i=1}^j p_i)^2$
	and, finally $C_j=\sum_{i=1}^j p_i \leq 2 C^*_j$ for each $j\in J$ -- this reasoning is 
	the same as in~\cite{SH96b}.
	For each scenario $S_i$, $i\in [K]$, it holds 
		$f(\pi,S_i)=\sum_{j\in J} C_j w_j(S_i)\leq 2 \sum_{j\in J} C^*_j w_j(S_i)$, and Observation~\ref{lemowa1} implies 
		$${\rm OWA}(\pi)={\rm owa}_{\pmb{v}}(\sum_{j\in J} C_j w_j(S_1), \dots, \sum_{j\in J} C_j w_j(S_K))\leq 2z^*.$$
		Since $z^*$ is a lower bound on the value of an optimal solution, $\pi$ is a 2-approximate schedule. Let us summarize the obtained result.
\begin{thm}
\label{thmcappr1}
	If $v_1\geq v_2\geq \dots\geq v_K$, and job processing times (or weights) are deterministic, then \textsc{Min-Owa}~$1|prec|\sum w_j C_j$ is approximable within~2.
\end{thm}
We now use Theorem~\ref{thmcappr1} to prove the following result:
\begin{thm}
	 \textsc{Min-Hurwicz}~$1|prec|\sum w_j C_j$ is approximable within~2, if job processing times (or weights) are deterministic.
\end{thm}
\begin{proof}
Assume that job processing times are deterministic (the reasoning for deterministic processing times is the same).
	The problem with the Hurwicz criterion can be rewritten as follows:
	$$ \min_{\pi \in \Pi} {\rm OWA}(\pi)= \min_{\pi \in \Pi}\min_{k\in [K]} H_k(\pi),$$
	where
	$$H_k(\pi)=\alpha \max_{i \in [K]} \sum_{j\in J} w_j(S_i)C_j(\pi) + (1-\alpha) \sum_{j\in J} w_j(S_k)C_j(\pi)=$$
	$$=\max_{i \in [K]}(\alpha \sum_{j\in J} w_j(S_i)C_j(\pi) + (1-\alpha) \sum_{j\in J} w_j(S_k)C_j(\pi))=\max_{i\in [K]} \sum_{j\in J} \hat{w}_j(S_i)C_j(\pi),$$
	where $\hat{w}_j(S_i)=\alpha w_j(S_i)+(1-\alpha)w_j(S_k)$. Hence the problem reduces to solving $K$ auxiliary 
	\textsc{Min-Max}~$1|prec|\sum w_j C_j$ problems. Since  
	\textsc{Min-Max}~$1|prec|\sum w_j C_j$  is approximable within~2 (see~\cite{MNO13}, or Theorem~\ref{thmcappr1}),
	 the theorem follows.
	\end{proof}
	
\section{Conclusion and open problems}
	
	In this paper we have proposed a new approach to scheduling problems with uncertain parameters. The key idea is to use the OWA operator to aggregate all possible values of the schedule cost. The weights in OWA allows decision makers to take their attitude towards a risk into account. In consequence, the main advantage of the proposed approach is to weaken the very conservative minmax criterion, traditionally used in robust optimization. Apart from proposing a general framework, we have discussed the computational properties of two basic single machine scheduling problems. We have shown that they have various computational and approximation properties, which make their analysis very challenging. However, there is still a number of open problems regarding the considered cases. For the problem with the maximum weighted tardiness criterion we do not know if the problem is weakly NP-hard when the number of scenarios is constant (the bounded case). It may be also the case that the pseudopolynomial algorithm designed for a fixed $K$ can be converted into a polynomial one by using a similar idea as for the Hurwicz criterion. We also do not know if the problem with the average criterion admits an approximation algorithm with a constant worst-case ratio (we only know that it is approximable within $K$ and not approximable within a ratio less than 7/6). For the problem with the weighted sum of completion times criterion the complexity of $\textsc{Min-Average}~1||\sum w_jC_j$ with uncertain processing times and weights is open.
The framework proposed in this paper can also be applied to other scheduling problems, for example to the single machine scheduling problem with the sum of late jobs criterion (the minmax version of this problem was discussed
	 in~\cite{AAK11,AC08}).
	 
%***********Acknowledgements***********************************
\subsubsection*{Acknowledgements}
This work was 
partially supported by
 the National Center for Science (Narodowe Centrum Nauki), grant  2013/09/B/ST6/01525.	 
	 
%\bibliographystyle{abbrv} 
%\bibliography{robseq} 

\begin{thebibliography}{10}

\bibitem{AAK11}
H.~Aissi, M.~A. Aloulou, and M.~Y. Kovalyov.
\newblock Minimizing the number of late jobs on a single machine under due date
  uncertainty.
\newblock {\em Journal of Scheduling}, 14:351--360, 2011.

\bibitem{AC08}
M.~A. Aloulou and F.~D. Croce.
\newblock Complexity of single machine scheduling problems under scenario-based
  uncertainty.
\newblock {\em Operations Research Letters}, 36:338--342, 2008.

\bibitem{AV01}
I.~Averbakh.
\newblock On the complexity of a class of combinatorial optimization problems
  with uncertainty.
\newblock {\em Mathematical Programming}, 90:263--272, 2001.

\bibitem{AZ02}
A.~Avidor and U.~Zwick.
\newblock Approximating {MIN $k$-SAT}.
\newblock {\em Lecture Notes in Computer Science}, 2518:465--475, 2002.

\bibitem{B07}
P.~Brucker.
\newblock {\em Scheduling Algorithms}.
\newblock Springer Verlag, Heidelberg, 5th edition, 2007.

\bibitem{GS12}
L.~Galand and O.~Spanjaard.
\newblock {E}xact algorithms for {OWA}-optimization in multiobjective spanning
  tree problems.
\newblock {\em Computers and Operations Research}, 39:1540--1554, 2012.

\bibitem{GLLK79}
R.~L. Graham, E.~L. Lawler, J.~K. Lenstra, and A.~H.~G. Rinnooy~Kan.
\newblock Optimization and {A}pproximation in {D}eterministic {S}equencing and
  {S}cheduling: {A} {S}urvey.
\newblock {\em Annals of Discrete Mathematics}, 5:169--231, 1979.

\bibitem{HA97}
L.~A. Hall, A.~S. Schulz, D.~B. Shmoys, and J.~Wein.
\newblock Scheduling to minimize average completion time: off-line and on-line
  approximation problems.
\newblock {\em Mathematics of Operations Research}, 22:513--544, 1997.

\bibitem{YKB11}
J.~Kacprzyk, R.~Yager, and G.~E. Beliakov.
\newblock {\em Recent Developments in the Ordered Weighted Averaging Operators:
  Theory and Practice}.
\newblock Studies in Fuzziness and Soft Computing, 265. Springer, 2011.

\bibitem{K05}
A.~Kasperski.
\newblock Minimizing maximal regret in the single machine sequencing problem
  with maximum lateness criterion.
\newblock {\em Operations Research Letters}, 33(4):431--436, 2005.

%\bibitem{KZ12a}
%A.~Kasperski, A.~Kurpisz, and P.~Zieli{\'n}ski.
%\newblock Parallel machine scheduling under uncertainty.
%\newblock In S.~Greco, B.~Bouchon-Meunier, G.~Coletti, M.~Fedrizzi,
%  B.~Matarazzo, and R.~R. Yager, editors, {\em Advances in Computational
%  Intelligence - 14th International Conference on Information Processing and
%  Management of Uncertainty in Knowledge-Based Systems, {IPMU} 2012,
%  Proceedings, Part {IV}}, volume 300 of {\em Communications in Computer and
%  Information Science}, pages 74--83. Springer, 2012.

\bibitem{KZ13}
A.~Kasperski, A.~Kurpisz, and P.~Zieli{\'n}ski.
\newblock Approximating the min-max (regret) selecting items problem.
\newblock {\em Information Processing Letters}, 113:23--29, 2013.

\bibitem{KZ11f}
A.~Kasperski and P.~Zieli{\'n}ski.
\newblock Possibilistic minmax regret sequencing problems with fuzzy
  parameters.
\newblock {\em IEEE Transactions on Fuzzy Systems}, 19:1072--1082, 2011.


\bibitem{KZ14}
A.~Kasperski and P.~Zieli{\'n}ski.
\newblock Minmax (regret) scheduling problems.
In Y.~Sotskov, F.~Werner, editors, {\em Sequencing and Scheduling with Inaccurate Data}, pages 159--210. Nova Science Pub., 2014.


\bibitem{KM94}
R.~Kohli, R.~Krishnamurti, and P.~Mirchandani.
\newblock The minimum satisfiability problem.
\newblock {\em SIAM Journal on Discrete Mathematics}, 7:275--283, 1994.

\bibitem{KY97}
P.~Kouvelis and G.~Yu.
\newblock {\em Robust Discrete Optimization and its applications}.
\newblock Kluwer Academic Publishers, 1997.

\bibitem{LA73}
E.~L. Lawler.
\newblock Optimal sequencing of a single machine subject to precedence
  constraints.
\newblock {\em Management Science}, 19:544--546, 1973.

\bibitem{LK78}
J.~K. Lenstra and A.~H.~G. Rinnooy~Kan.
\newblock Complexity of scheduling under precedence constraints.
\newblock {\em Operations Research}, 26:22--35, 1978.

\bibitem{LR57}
R.~D. Luce and H.~Raiffa.
\newblock {\em {G}ames and {D}ecisions: {I}ntroduction and {C}ritical
  {S}urvey}.
\newblock Dover Publications Inc., 1957.

\bibitem{MR96}
M.~V. Marathe and S.~S. Ravi.
\newblock On approximation algorithms for the minimum satisfiability problem.
\newblock {\em Information Processing Letters}, 58:23--29, 96.

\bibitem{MNO13}
M.~Mastrolilli, N.~Mutsanas, and O.~Svensson.
\newblock Single machine scheduling with scenarios.
\newblock {\em Theoretical Computer Science}, 477:57--66, 2013.

\bibitem{OS03}
W.~Ogryczak and T.~{\'S}liwi{\'n}ski.
\newblock On solving linear programs with ordered weighted averaging objective.
\newblock {\em European Journal of Operational Research}, 148:80--91, 2003.

\bibitem{P02}
M.~Pinedo.
\newblock {\em Scheduling: Theory, Algorithms, and Systems}.
\newblock Prentice Hall, 2002.

\bibitem{POT80}
C.N.~Potts.
\newblock An algorithm for the single machine sequencing problem with precedence constraints.
\newblock {\em Mathematical Programming Study}, 13:78--87, 1980.

\bibitem{FA10}
I.~Regis~de Farias, H.~Zhao, and M.~Zhao.
\newblock A family of inequalities valid for the robust single machine
  scheduling polyhedron.
\newblock {\em Computers and Operations Research}, 37:1610--1614, 2010.

\bibitem{SH96a}
A.~S. Schulz.
\newblock {\em {P}olytopes and {S}cheduling}.
\newblock PhD thesis, Technical University of Berlin, Germany, 1996.

\bibitem{SH96b}
A.~S. Schulz.
\newblock Scheduling to minimize total weighted completion time: Performance
  guarantees of {LP-Based} heuristics and lower bounds.
\newblock In {\em IPCO}, pages 301--315, 1996.

\bibitem{VD10}
A.~T. Volgenant and C.~W. Duin.
\newblock Improved polynomial algorithms for robust bottleneck problems with
  interval data.
\newblock {\em Computers and Operations Research}, 37:909--915, 2010.

\bibitem{YA88}
R.~R. Yager.
\newblock On ordered weighted averaging aggregation operators in multi-criteria
  decision making.
\newblock {\em IEEE Transactions on Systems, Man and Cybernetics}, 18:183--190,
  1988.

\bibitem{YY02}
J.~Yang and G.~Yu.
\newblock On the robust single machine scheduling problem.
\newblock {\em Journal of Combinatorial Optimization}, 6:17--33, 2002.

\end{thebibliography}

\end{document}